\documentclass[10pt, conference]{IEEEtran}  
\usepackage{tikz}
\usetikzlibrary{shapes,snakes}
\usepackage{empheq}
\usepackage{acronym}
\usepackage{graphicx}
\usepackage{epstopdf}
\usepackage{multirow}
\usepackage{tabu}
\usepackage{array}
\usepackage{url}
\usepackage{amsmath}
\usepackage{amssymb}
\usepackage{amsthm} 
\usepackage{amsfonts}
\usepackage{tikz}
\usepackage{bm}
\usepackage{algorithmic}
\usepackage[ruled,vlined,linesnumbered]{algorithm2e}
\usepackage{graphicx}
\usetikzlibrary{arrows}
\usepackage{cite}
\usepackage{caption}
\usepackage{color}
\usepackage{subcaption}
\usepackage{amssymb}
\usepackage{tabulary}
\usepackage{booktabs}

\usepackage[justification=centering]{caption} 
\DeclareCaptionLabelFormat{algnonumber}{algorithm}
\newtheorem{theorem}{Theorem}
\newtheorem{lemma}{Lemma}

\tikzstyle{int}=[draw, fill=white!20, minimum size=2em]
\tikzstyle{init} = [pin edge={to-,thin,black}]

\DeclareRobustCommand*{\IEEEauthorrefmark}[1]{%
	\raisebox{0pt}[0pt][0pt]{\textsuperscript{\footnotesize\ensuremath{#1}}}}

\newcommand{\ve}[1]{\boldsymbol{\mathbf{#1}}}
\newcommand{\hi}{\text{High}}
\newcommand{\lo}{\text{Low}}
\newcommand{\set}[1]{\mathcal{#1}}

\acrodef{AoI}{Age of Information}
\acrodef{IoT}{Internet of Things}
\acrodef{DPC}{Dynamic Power Control}
\acrodef{LDF}{Largest-Debt-First}
\acrodef{VR}{Virtual Reality}
\acrodef{eMBB}{enhance Mobile Brodband}
\acrodef{URLLC}{Ultra Reliable Low Latency}
\acrodef{TTI}{Transmission Time Interval}
\acrodef{QoS}{Quality of Service}

\IEEEoverridecommandlockouts                          

\title{
	Dynamic Power Control for Time-Critical Networking with Heterogeneous Traffic
}

\author{
	\IEEEauthorblockN{Emmanouil Fountoulakis\IEEEauthorrefmark{1},
		Nikolaos Pappas\IEEEauthorrefmark{1},
		Anthony Ephremides \IEEEauthorrefmark{1,2}
		\IEEEauthorblockA{\IEEEauthorrefmark{1} Department of Science and Technology, Link{\"o}ping University, Sweden}
		\IEEEauthorblockA{\IEEEauthorrefmark{2} Electrical and Computer Engineering Department, University of Maryland, College Park}
		E-mails: \{emmanouil.fountoulakis, nikolaos.pappas\}@liu.se, etony@umd.edu}
}

\begin{document}
	
	\maketitle
	
	\thispagestyle{empty}
	\pagestyle{empty}
	
	\begin{abstract}
		Future wireless networks will be characterized by heterogeneous traffic requirements. 
		Such requirements can be low-latency or minimum-throughput. Therefore, the network has to adjust to  different needs. Usually, users with low-latency requirements have to deliver their demand within a specific time frame, i.e., before a deadline, and they co-exist with throughput oriented users. In addition, the users are mobile and they share the same wireless channel. Therefore, they have to adjust their power transmission to achieve reliable communication. However, due to the limited-power budget of wireless mobile devices, a power-efficient scheduling scheme is required by the network. In this work, we cast a stochastic network optimization problem for minimizing the packet drop rate while guaranteeing a minimum-throughput and taking into account the limited-power capabilities of the users. We apply tools from Lyapunov optimization theory in order to provide an algorithm, named \ac{DPC} algorithm, that solves the formulated problem in real-time. It is proved that the \ac{DPC} algorithm gives a solution arbitrarily close to the optimal one. Simulation results show that our algorithm outperforms the baseline \ac{LDF} algorithm for short deadlines and multiple users.
	\end{abstract}
	
	
	\IEEEpeerreviewmaketitle 
	
	\section{Introduction}

	5G and beyond networks are poised to support a mixed set of applications that require different types of services. There are two main categories of applications. The first category includes applications that require bandwidth-hungry services and the second includes delay-sensitive applications.
	The second category differentiates the current networks from future networks. These applications require low-latency services and increase the need for time-critical networking.  
	In time-critical networking, applications require to deliver their demands within a specific time duration \cite{ericssosreview2020}. In other words, each packet or a batch of packets has a deadline within which must be transmitted, otherwise, it is dropped and removed from the system \cite{hou2013packets}. This is connected with the notion of \textit{timely throughput}. Timely throughput measures the long-term time average number of successful deliveries before the deadline expiration \cite{TheoryQoS2009, lashgari2013timely}.
	Each time-critical application belongs to a different category. For example, motion control, smart grid control, and process monitoring belong to the industrial control category. Furthermore, the growing popularity of real-time media applications increases the need for designing networks that can offer services with low latency. Such applications are media production, interactive \ac{VR}, cloud computing, etc, that are under the umbrella of the Tactile Internet \cite{holland2019ieee}.
	
	With the pervasiveness of mobile communications, such applications need to perform over wireless devices. In order to achieve reliable communication, the devices have to adapt their power transmission according to channel conditions. However, many devices may have a limited power budget. Therefore, energy-efficient communications have become a very important issue. In this work, we propose a scheduling algorithm that handles a heterogeneous set of users with heterogeneous traffic. In particular, we consider a network with deadline-constrained users and users with minimum-throughput requirements, with a limited-power budget. We provide an algorithm that solves the scheduling problem in real-time. We prove that the obtained solution is arbitrarily close to the optimal.
	
	\subsection{Related works}
	Delay-constrained network optimization and performance analysis have been extensively investigated  \cite{SurveyDelayAwareResourceAllocation}. A variety of approaches have been applied to different scenarios. There is a line of works that consider the control of the maximum number of retransmissions before the deadline expiration. In \cite{BambosPerfEvalu}, the authors consider a user transmitting packets over a wireless channel to a receiver. An optimal scheduling scheme is proposed that provides the optimal number of retransmission for a packet. In \cite{nomikos2018deadline}, the authors consider users with packets with deadlines in a random-access network. They show how the number of maximum retransmissions affects the packet drop rate. In \cite{faridi2008distortion}, the authors consider a single transmitter that transmits symbols to a receiver. Each symbol has a deadline and a corresponding distortion function. The authors consider the distortion-minimization problem while fulfilling deadline constraints.
	
	Furthermore, energy and power efficient scheduling schemes for  delay-constrained traffic have attracted a lot of attention over last the years \cite{fountoulakis2018dynamic, choi2019dynamic,AnOnlineAlgorBorkar, GoyalPowerConstrained2003,BambosDeadlines,OptimalTransmissionSchedulingTsitsiklis}. In \cite{fountoulakis2018dynamic}, the authors consider the minimization of drop rate for users with limited power-budget. They propose an approximated algorithm that performs in real time. In \cite{choi2019dynamic}, the authors propose an algorithm that minimizes the time average power consumption while guaranteeing minimum throughput and reducing the queueing delay. They also consider a hybrid system multiple access system where the scheduler decides if the transmitter serves a user by by orthogonal multiple access or non-orthogonal multiple access. In \cite{AnOnlineAlgorBorkar, GoyalPowerConstrained2003}, the authors utilize markov decision theory to provide an optimal energy-efficient algorithm for delay-constrained user. 
	
	Lyapunov optimization theory has been widely applied  for developing dynamic algorithms that schedule users with packets with deadlines. In \cite{PowerControlTeped,ewaisha2017optimal, ElAzzouni2020,kim2015optimal}, the authors consider the rate maximization under power and delay constraints. In \cite{PowerControlTeped}, the authors consider the power allocation for users with hard-deadline constraints. In \cite{ewaisha2017optimal}, the author consider the rate maximization of non-real time users while satisfying packet drop rate for users with packets with deadlines. 
	In \cite{ElAzzouni2020,kim2015optimal}, the authors consider packet with deadlines for scheduling real-time traffic in wireless environment.   
	A novel approach for minimizing packet drop rate while guaranteeing stability is provided in \cite{neely2013dynamic}. The authors combine tools from Lyapunov optimization theory and markov decision processes in order to develop an optimal algorithm for minimizing drop rate under stability constraints. However, the algorithm is able to solve small network scenarios because of the curse of dimensionality problem.

	Besides delay-constrained traffic management, throughput-optimal algorithms have been developed over the years. Following the seminal work in \cite{tassiulas1990stability}, many researchers developed different solutions for the throughput-maximization problem by proposing a variety of approaches \cite{li2014throughput, lan2019throughput,sadiq2009throughput}. In \cite{li2014throughput}, the authors consider the throughput-maximization while guaranteeing certain interservice times for all the links. They propose the time-since-last-service metric. They combine the last with the queue length of each user and they propose a max-weight policy based on Lyapunov optimization. In \cite{lan2019throughput,sadiq2009throughput}, the authors consider the throughput-maximization in networks with dynamic flows. More specifically, in \cite{lan2019throughput}, the authors consider a hybrid system with both persistent and dynamic flows. They provide a queue-maximum-weight based algorithm that guarantees throughput-optimality while reducing the latency. In \cite{sadiq2009throughput}, the authors consider a network with dynamic flows of random size and they arrive in random size at the base station. The service times for each flow varies randomly because of the wireless channel. The authors provide a delay-MaxWeight scheduler that is proved that is throughput-optimal. 
	
	Research on scheduling heterogeneous traffic with \ac{URLLC} users and \ac{eMBB} has attracted a lot of attention by the community \cite{you2018resource, TTIWiOpt,anand2020joint, anand2018resource,karimi2020low,khalifa2020low,avranas2019throughput,destounis2018scheduling}.
	In \cite{you2018resource,TTIWiOpt}, the authors show the benefits of flexible \ac{TTI} for scheduling users with different types of requirements. In \cite{anand2020joint}, the authors propose an algorithm that jointly schedules \ac{URLLC} and \ac{eMBB} traffic. They consider a slotted time system in which the slots are divided into mini-slots. They consider the frequency and mini-slots allocation over one slot. In \cite{anand2018resource}, the authors consider the resource allocation for \ac{URLLC} users. They study resource allocation for different scenarios: i) OFDMA system, ii) system that includes re-transmissions. In 
	\cite{karimi2020low,khalifa2020low},  the authors propose low-complexity algorithm for scheduling \ac{URLLC} users. The authors in\cite{avranas2019throughput} consider the throughput maximization and HARQ optimization for \ac{URLLC} users. Furthermore, reliable transmission is an important issue of \ac{URLLC} communications. In \cite{destounis2018scheduling}, the authors consider a network in which multiple unreliable transmissions are combined to achieve reliable latency. The authors model the problem as a constrained Markov decision problem, and they provide the optimal policy that is based on dynamic programming.
	
	\subsection{Contributions}
	In this work, we consider two sets of users with heterogeneous traffic and a limited-power budget. The first set includes users with packets with deadlines and the second set includes users with minimum-throughput requirements. We provide a dynamic algorithm that schedules the users in real-time and minimizes the drop rate while guaranteeing minimum-throughput and limited-power consumption. The contributions of this work are the following. 
	\begin{itemize}
		\item We formulate an optimization problem for minimizing the drop rate with minimum-throughput constraints and time average power consumption constraints. 
		\item We provide a novel objective function for minimizing the drop rate. The objective function does not take into account only if a packet is going to expire or not, but also the remaining time of a packet before its expiration.
		\item We apply tools from the Lyapunov optimization theory to satisfy the time average constraints: throughput and power consumption.
		\item We prove that the proposed algorithm always satisfies the long term constraints if the problem is feasible.
		\item The proposed algorithm is proved to provide a solution arbitrarily close to the optimal.
		\item Simulation results show that our algorithm outperforms the baseline algorithm proposed in \cite{TheoryQoS2009} for short deadlines and multiple users.
	\end{itemize}

	\section{System Model}\label{Sec:SytemModel}
	\begin{figure}[h!]
		\centering
		\includegraphics[scale=0.37]{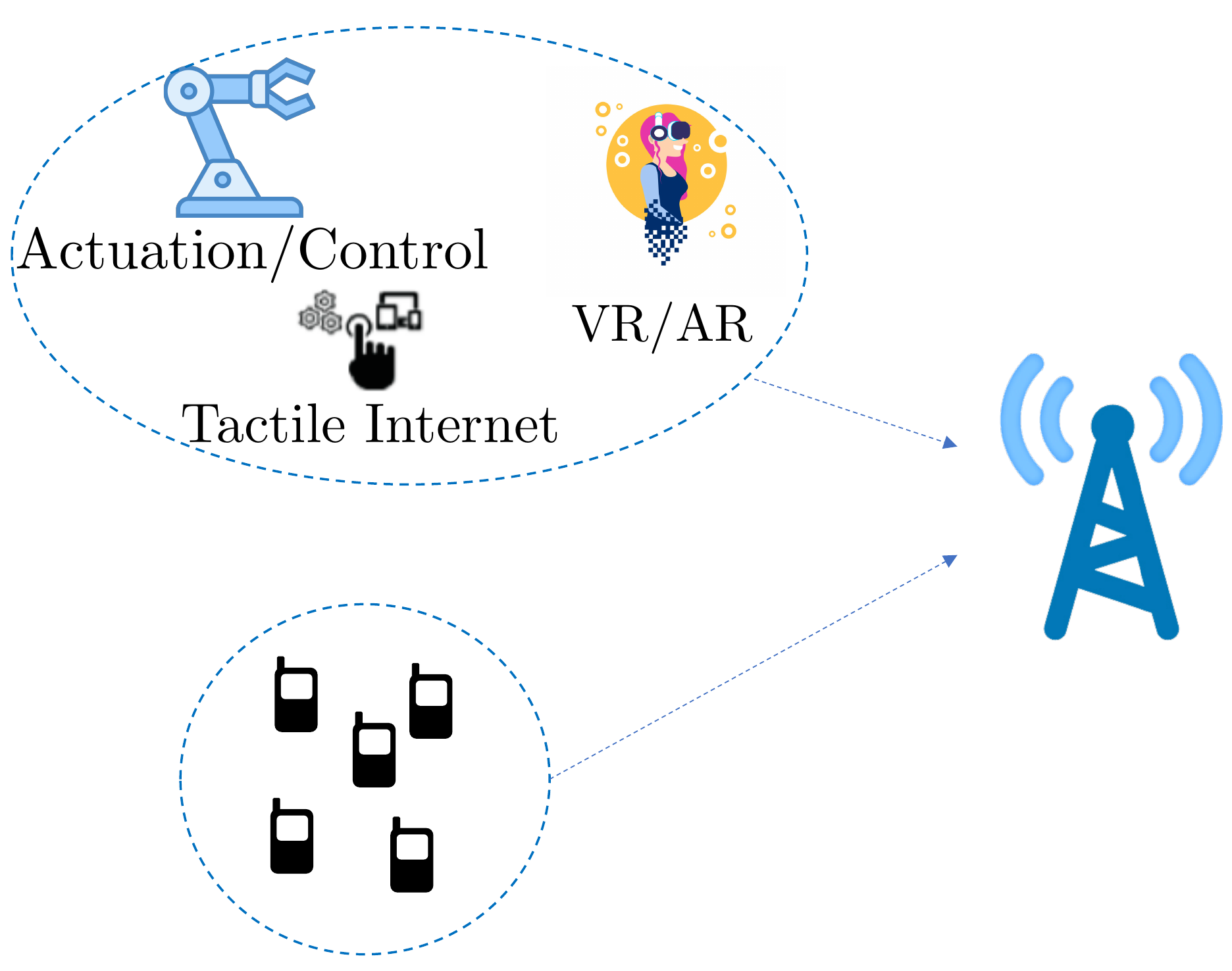}
		\caption{Example of the system model with time-critical users with packets with deadlines and users with minimum-throughput requirements.}
		\label{fig:systemmodel}
	\end{figure}
	\vspace{1mm}
	We consider $N$ users transmitting packets to a single receiver over wireless fading channels. Let $\mathcal{N}\triangleq \left\{1,\dots,N\right\}$ be the set of the total users in the system. Time is assumed to be slotted, and $t \in \mathbb{Z}$ denotes the $t^{\text{th}}$ slot. 
	\begin{table*}[t!]\caption{Notation Table.}
		\begin{center}
			\begin{tabular}{r c l l p{5cm}}
				\toprule
				&$\mathcal{N}$    & Set of the total users in the system                           	&  $\mathcal{P}_{i}(t)$     & Set of selectable power levels of user $i$\\
				&$\mathcal{R}$    & Set of the deadline-constrained users							 &          $Z_u(t)$			  & Length of throughput-dept queue of user $u$\\
				&$\mathcal{U}$    & Set of the minimum-throughput requirements users  	 &          $f_r(t)$ 				& Cost function for user $r$\\
				&$t$              & $t^{\text{th}}$ slot                                 								 &       $\mu_{r}(t)$            & Power allocation indicator of user $i$\\
				&$Q_{r}(t)$       & Number of packets in queue $r$                       &  $\mathcal{P}(t)$         & Set of power constraints for $\mathbf{p}(t)$\\
				&$\pi_{r}$          & Packet arrival probability of user $i$                      &  $D_{r}(t)$               & Packet drop indicator of user $r$\\
				&$\alpha_{r}(t)$  & Packet arrival indicator of user $r$                 &  $\overline{D}_{r}$       & Packet drop rate of user $r$\\
				&$m_{r}$          & Deadline of packet of user $r$                       &  $\overline{p}_{r}$       & Average power consumption of user $r$\\
				&$d_{r}(t)$       & Number of slots left before the deadline of user $r$ & $X_{r}(t)$                & Length of virtual queue of user $i$\\
				&$\mathbf{S}(t)$  & Channel states                                       & $L(\cdot)$                & Quadratic Lyapunov function \\
				& $\mathbf{p}(t)$ & Power allocation vector                              & $\Delta(L(\cdot))$        & Lyapunov drift\\
				& $\gamma_{i}$    & Allowed average power consumption for user $i$       & $\bm{\alpha}(t)$             & Packet arrival indicator vector\\
				\bottomrule
			\end{tabular}
		\end{center}
		\label{tab:TableOfNotationForMyResearch}
	\end{table*}

	We consider the users to be synchronized and controlled by a scheduler. In our system model, at most one user can transmit at each time slot.
	We consider two sets of users. The first set includes users that have arrivals of packets with deadlines. We denote the set of users with packets with deadlines by $\mathcal{R} \subseteq \mathcal{N}$. The second set of users, $\mathcal{U} \subseteq \mathcal{N}$, contains the users that have throughput requirements. We consider that each user $u \in \mathcal{U}$ is saturated and therefore, it has always packet to transmit. Note that $\mathcal{U}\cup \mathcal{R} = \mathcal{N}$ and $\mathcal{U}\cap \mathcal{R} = \emptyset$. An example of our system model is shown in Fig. \ref{fig:systemmodel}.
	
	For the deadline-constrained users, each packet that arrives in the queue of the users has a deadline by which the packet must be transmitted. Otherwise, it is dropped and removed from the system. We assume that the  deadlines of the packets in the same queue are the same. However, deadlines of different queues may vary. We denote the packet deadline of the $r^{\text{th}}$ queue with $m_{r} \in \mathbb{Z}_{+}\text{, }\forall r \in \mathcal{R}$.  Let $d_{r}(t)$ be the number of slots left in the $t^{\text{th}}$ slot before the packet that is at the head of queue $r$ expires.
	Let $Q_{r}(t)$ be the number of packets in queue $r$ in the $t^{\text{th}}$ slot. A packet arrives with probability $\pi_r$ at every time slot in the queue of  user $r$. Let $\bm{\alpha}(t)\triangleq \left\{\alpha_{r}(t)\right\}_{r\in \mathcal{R}}$, where $\alpha_{r}(t) \in \left\{0,1\right\}$, represents the packet arrival process for each user $r$ in the $t^{\text{th}}$ slot. The random variables of packet arrival process are independent and identically  distributed (i.i.d.). Let $\lambda_r$ denote the arrival rate for user $r$ and $\mathbb{E}\{\alpha_r(t)\} = \lambda_r$.  Furthermore, at most one packet can be transmitted at each time slot and no collisions are allowed. In each queue of every user $r \in \mathcal{R}$, packets are served in the order that they arrive following the First In First Out (FIFO) discipline.

	We assume that the channel state at the beginning of each time slot is known. The channel state remains constant within one slot but it changes from slot to slot. Let $\mathbf{S}(t)\triangleq \left\{S_{i}(t)\right\}_{i\in \mathcal{N}}$ represent the channel state for each user $i$ during slot $t$. Also the channel can be either in ``Bad" state (deep fading) or in ``Good" state (mild fading). The possible channel states of each user $i$ are described by the set $\mathcal{S}\triangleq\left\{\text{B},\text{G}\right\}$, and $S_{i}(t)\in \mathcal{S}$, $\forall i \in \mathcal{N}$. For simplicity, we assume that the random variables of the channel process $\mathbf{S}(t)$ are i.i.d. from one slot to the next. 
	
	Let $\ve{p}(t)\triangleq \left[p_{1}(t),\ldots,p_{N}(t)\right]$ denote the power allocation vector in the $t^{\text{th}}$ slot. We consider a set of discrete power levels $\left\{0, P^{(\lo)}, P^{(\hi)}\right\}$.
	The required power to have a successfull transmission under ``Bad" and ``Good" channel condition is denoted by $P^{\text{High}}$ and $P^{\text{Low}}$, respectively.
	At each time slot, the set of selectable power levels $\set{P}_{i}(t)$ for each user is conditioned on the channel state $S_i(t)$. For example, if the current channel state is ``Bad", then $P^{\text{(Low)}}$ cannot be selected. Thus, we have 
	\begin{align}\label{IndividualPowSet}
	p_{i}(t) \in \begin{cases}
	\left\{0, P^{\text{(High)}}\right\}\text{, if } S_i(t) = \text{B}\\
	\left\{0, P^{\text{(Low)}}\right\}\text{, if } S_i(t) = \text{G}
	\end{cases}\text{, }\forall i\in\set{N}\text{.}
	\end{align}
	Let $\mu_{i}(t)$ be the power allocation, or packet serving, indicator for the user $i$ in the $t^{\text{th}}$ slot, we have
	\begin{align}\label{AllocIndic}
	\mu_{i}(t) \triangleq \begin{cases}
	1\text{, if } p_i(t)>0\\
	0\text{, otherwise}
	\end{cases}\text{, }\forall i\in\set{N}\text{.}
	\end{align}
	%
	At most one packet can be transmitted in a timeslot $t$, i.e., the vector $\mathbf{p}(t)$ has at most one non-zero element. The set of power constraints for $\ve{p}(t)$ is then defined by 
	\begin{equation}\label{eqn:GeneralPowSet}
	\set{P}(t)\triangleq \left\{\ve{p}(t): \sum_{i=1}^N \mathbf{1}_{\{\mu_i(t)=1\}}\leq 1\right\}\text{,}
	\end{equation}
	where $\mathbf{1}_{\{\cdot\}}$ denotes the indicator function. 
	For each user $r \in \mathcal{R}$, a packet is dropped if its deadline has expired. Since the queue follows FIFO discipline, a packet is dropped under the following conditions: 1) it is at the head of the queue; 2) the remaining number of the slots to serve the packet is $1$; and 3) power is not assigned to user $r$ at the current slot. Let $D_{r}(t)$ be the indicator of the packet drop for user $r$ at time $t$.
	The queue evolution for each user $r \in \mathcal{R}$ is described as
	\begin{align}\label{QueueEvolutionDeadline}
	Q_{r}(t+1) = \max \left[Q_{r}(t)- \mu_{r}(t),0 \right] 
	+ \alpha_{r}(t)-D_{r}(t)\text{, } \forall i \in \mathcal{R}\text{.}
	\end{align}
	\vspace{1mm}
	We define the packet drop rate for each user $r$ $\in$ $\mathcal{R}$,  the average power consumption for each user $i \in \mathcal{N}$, and the throughput for each user $u \in \mathcal{U}$ as
	\begin{align}
	\overline{D}_{r} & \triangleq \lim\limits_{t\rightarrow\infty}  \overline{D}_{r}(t)\text{, } \forall r \in \mathcal{R}\text{,}\label{DropRate}\\
	\overline{p}_{i} & \triangleq \lim\limits_{t\rightarrow \infty} \overline{p}_{i}(t)\text{, }
	\forall i \in \mathcal{N}\text{,}\label{def: avepower}\\
	\overline{\mu}_u & = \lim\limits_{t\rightarrow \infty}  \bar{\mu}_u (t)\text{, } \forall u\in \mathcal{U}\label{def: throughtput}\text{,}
	\end{align}
	respectively, where, $\overline{D}_{r}(\tau)=\frac{1}{t}\sum\limits_{\tau=0}^{t-1} D_{r}(t)$, $\overline{p}_{i}(t)=\frac{1}{t}\sum\limits_{\tau=0}^{t-1} p_{i}(\tau)$, and $\bar{\mu}_u(t) = \frac{1}{t}\sum\limits_{\tau=0}^{t-1}\mu_u(\tau)$. The packet drop rate represents the average number of dropped packets per time slot. The average power consumption represents the average of transmit power over all time slots. The throughput represents the average served packets per time slot for each user $u \in \mathcal{U}$.
	
	These metrics are connected and we will show in the following sections how the average power consumption affects the packet drop rate and the throughput.

	\section{Problem Formulation}\label{sec:formulation}
	Our goal is to achieve the minimum drop rate for deadline-constrained users while providing a minimum throughput for each user $u \in \mathcal{U}$ and keeping the average power consumption for every user below a threshold. To this end, we provide the following stochastic optimization problem
	\begin{subequations}
		\begin{align}\label{Opt.Probl}
		\min\limits_{\bm{p}(t)} \quad & \sum\limits_{r\in\mathcal{R}} \overline{D}_{r}\\\label{PowerConstr}
		\text{s.~t.} \quad & \overline{p}_{i}\leq \gamma_{i}\text{, } \forall i \in  \mathcal{N}\text{,}\\
		\quad &\bar{\mu}_u\geq \delta_u\text{, } u\in\mathcal{U}\label{ThrConst}\text{,}\\
		\quad & \bm{p}(t) \in \mathcal{P}(t)\text{,}
		\end{align}
	\end{subequations}
	where $\gamma_{i} \in \left[0,P^{\text{(High)}}\right]$ indicates the allowed average power consumption for each user $i$. Also, $\delta_u$ denotes the minimum throughput requirement for each user $u \in \mathcal{U}$. The constraint in (\ref{PowerConstr}) ensures that average power consumption of each user $i$ remains below $\gamma_{i}$ power units.
	
	The formulation above represents our intended goal which is the minimization of the packet drop rate. However, the objective function in (\ref{Opt.Probl})  makes the solution approach non-trivial. The decision variable, $\mathbf{p}(t)$ (power allocation), is optimized slot-by-slot for minimization of the objective function that is defined over infinite time horizon. We have to cope with one critical point: we do not have prior knowledge about the future states of the channel and packet arrivals in the system. Therefore, we are not able to predict the values of the objective function in the future slots in order to decide on the power allocation that minimizes the cost. We aim to design a function whose future values are affected by the current decision and the remaining expiration time of the packets. 
	To this end, we introduce a function incorporating the relative difference between the packet deadline $m_{r}$ and the number of remaining future slots $(d_{r}(t)-1)$ before its expiration as described below
	\begin{equation}\label{NewFunction}
	f_{r}(t) \triangleq \frac{m_{r}-(d_{r}(t)-1)}{m_{r}} \mathbf{1}_{\left\{\mu_{r}(t)=0\right\}}\text{, } \forall r \in \mathcal{R}\text{.}
	\end{equation}
	The function in (\ref{NewFunction}) takes its extreme value, $f_{r}(t)=0$, when a packet of user $r \in \mathcal{R}$ is served, or $f_{r}(t)=1$ when a packet of user $r \in \mathcal{R}$ is dropped. Therefore, that function takes the same values with those of (\ref{DropRate}) in the extreme cases. In addition, the function in (\ref{NewFunction}) assigns the cost according to the remaining time of a packet to expire  in the intermediate states, i.e., when a packet is waiting in the queue. The cost increases when there is less time left for serving the packet with respect to the defined deadline. The time average of $f_{r}(t)$ is 
	\begin{align}\label{eqn:NewObj}
	\overline{f_{r}} \triangleq \lim\limits_{t\rightarrow \infty} \overline{f}_{r}(t)\text{, } \forall r \in \mathcal{R}\text{,}
	\end{align}
	where $\overline{f}_{r}(t)\triangleq \frac{1}{t} \sum\limits_{\tau=0}^{t-1} f_r(\tau)$ and
	\begin{align}
	f  =  \sum\limits_{r \in \mathcal{R}} f_r(t)\text{.}
	\end{align}
	Finally, we formulate a minimization problem by using \eqref{eqn:NewObj} as shown below
	\begin{subequations}
		\label{Eq.Probl}
		\begin{align}
		\min\limits_{\bm{p}(t)} \quad & \sum\limits_{r\in\mathcal{R}} \overline{f}_{r}\\\label{PowerConstrEq}
		\text{s.~t.} \quad & \overline{p}_{i}\leq \gamma_{i}\text{, } \forall i \in  \mathcal{N}\text{,}\\
		\quad &\bar{\mu}_u\geq \delta_u\text{, } u\in\mathcal{U}\label{ThrConstEq}\text{,}\\
		\quad & \bm{p}(t) \in \mathcal{P}(t)\text{.}
		\end{align}
	\end{subequations}

	\section{Proposed Approximate Solution} 
	The problem in (\ref{Eq.Probl}) includes time average constraints. In order to satisfy these constraints, we aim to develop a policy that uses techniques different from standard optimization methods based on static and deterministic models. Our approach is based on Lyapunov optimization theory \cite{NeelyBook}.
	
	In particular, we apply the technique developed in \cite{NeelyEnergyOptimalControl} and further discussed in \cite{NeelyBook} and \cite{TassiulasNeelyNow} in order to develop a policy that ensures that the constraints in \eqref{PowerConstrEq} and \eqref{ThrConstEq}  are satisfied. 
	
	Each inequality constraint in  \eqref{PowerConstrEq} and \eqref{ThrConstEq} is mapped to a virtual queue. We show below that the power constraint and minimum throughput constraints problems are transformed into  queue stability problems. 
	
	
	Before describing the motivation behind the mapping of average constraints in \eqref{PowerConstrEq} and \eqref{ThrConstEq} to virtual queues, let us recall one basic theorem that comes from the general theory of stability of stochastic processes \cite{MarkovChains}. Consider a system with $K$ queues. The number of unfinished jobs of queue $i$ is denoted by $q_{k}(t)$, and
	$\mathbf{q}(t) = \left\{q_{k}(t)\right\}_{k=1}^{K}$. The Lyapunov function and the Lyapunov drift are denoted by $L(\mathbf{q}(t))$ and
	$\Delta(L(\mathbf{q}(t))) \triangleq E\left\{L(\mathbf{q}(t+1))-L(\mathbf{q}(t)) | \mathbf{q}(t)\right\}$ respectively \cite{MarkovChains}.
	Below we provide the definition of the Lyapunov function \cite{MarkovChains}.
	
	\textit{Definition 1 (Lyapunov function):} A function $L: \mathbb{R}^{K}\rightarrow \mathbb{R}$ is said to be a Lyapunov function if it has the following properties
	\begin{itemize}
		\item $L(\mathbf{x})\geq 0\text{, } \forall \mathbf{x} \in \mathbb{R}^{K}$,
		\item It is non-decreasing in any of its arguments,
		\item $L(\mathbf{x}) \rightarrow + \infty$, as $||\mathbf{x}||\rightarrow + \infty$.
	\end{itemize}
	\begin{theorem}[Lyapunov Drift]
		\label{ThLyapunov}
		If there exist positive values $B$, $\epsilon$ such that for all time slots $t$ we have
		$\Delta (L(\mathbf{q}(t)) \leq B - \epsilon \sum\limits_{k=1}^{K} q_{k}(t)\text{,}$
		then the system $\mathbf{q}(t)$ is \textit{strongly} stable. 
	\end{theorem}
	The intuition behind Theorem 1 is that if we have a queueing system, and we provide a scheduling scheme such that the Lyapunov drift is bounded and the sum of the length of the queues are multiplied by a negative value, then the system is stable.
	Our goal is to find a scheduling scheme for which the inequality of Theorem 1 holds for our application.
	
	Let $\left\{X_{i}(t)\right\}_{i\in \mathcal{N}}$ and $\left\{Z_{u}(t)\right\}_{u\in \mathcal{U}}$ be the virtual queues associated with constraints \eqref{PowerConstrEq} and \eqref{ThrConstEq}, respectively. We update each virtual queue $X_i(t)$ at each time slot $t$ as
	\begin{align}
	X_{i}(t+1) = \max\left[X_{i}(t)-\gamma_{i}, 0\right] + p_{i}(t)\text{,}
	\label{Xevolution}
	\end{align}
	and each virtual queue $Z_u(t)$ as
	\begin{align}
	Z_{u}(t+1) = \max\left[Z_u(t)-\mu_u(t),0\right] + \delta_u\text{.}
	\label{Zevolution}
	\end{align}
	Process $X_{i}(t)$  can be viewed as a queue  with ``arrivals" $p_{i}(t)$ and ``service rate" $\gamma_{i}$. Process $Z_u(t)$ can be also viewed as a queue with ``arrivals" $\delta_u$ and ``service rate" $\mu_u(t)$.
	
	We will show that the average constraints in \eqref{PowerConstrEq} and \eqref{ThrConstEq} are transformed into queue stability problems. Then, we develop a dynamic algorithm and we prove that the algorithm satisfies Theorem \ref{ThLyapunov} and achieves stability.
	\begin{lemma}
		\label{Th:rateStable}
		If $X_{i}(t)$ and $Z_{u}(t)$ are \textit{rate stable}\footnote{A discrete time process $Q(t)$ is \textit{rate stable} if $\lim\limits_{t\rightarrow \infty}\frac{Q(t)}{t}=0$ with probability $1$ \cite{NeelyBook}.}, then the constraints in  \eqref{PowerConstrEq} and \eqref{ThrConstEq} are satisfied.
	\end{lemma}
	\begin{proof}
		See Appendix B.
	\end{proof}
	
	Note that strong stability implies all of the other forms of stability \cite[Chapter 2]{NeelyBook} including the rate stability. Therefore, the problem is transformed into a queue stability problem.
	In order to stabilize the virtual queues $X_{i}(t)\text{, } \forall i \in \mathcal{N}$ and $Z_u(t)\text{, } \forall u \in \mathcal{U}$, we first define the Lyapunov function as
	\begin{align}
	\label{Lfunction}
	L(\mathbf{\Theta}(t)) = \frac{1}{2} \sum\limits_{i\in\mathcal{N}} X^2_{i}(t) + \frac{1}{2}\sum\limits_{u\in\mathcal{R}}Z^2_u(t)\text{,}
	\end{align}
	where $\mathbf{\Theta}(t) = \left[ \{X_i(t)\}_{i\in\mathcal{N}}\text{, } \{Z_u(t)\}_{u\in\mathcal{U}}\right]$, and the Lyapunov drift as
	\begin{align}
	\label{Drift}
	\Delta(\mathbf{\Theta}(t)) \triangleq \mathbb{E}\left\{ L(\mathbf{\Theta}(t+1)) - L(\mathbf{\Theta}(t)) | \mathbf{\Theta}(t)  \right\}.
	\end{align}
	The above conditional expectation is with respect to the random channel states and the arrivals. 
	
	To minimize the time average of the desired cost $f_{r}(t)$ while stabilizing the virtual queues $X_{i}(t)$, $\forall i \in \mathcal{N}$, $Z_u(t)\text{, } \forall u \in \mathcal{U}$, we use the \textit{drift-plus-penalty} minimization approach introduced in \cite{TassiulasNeelyNow}. The approach seeks to minimize an upper bound on the following drift-plus-penalty expression at every slot $t$
	\begin{align}
	\label{driftPenaltyExpre}
	\Delta(\mathbf{\Theta}(t)) + V\sum\limits_{r \in \mathcal{R}} \mathbb{E}\left\{f_{r}(t)|\mathbf{\Theta}(t)\right\}\text{,}
	\end{align}
	where $V>0$ is an ``importance" weight to scale the penalty. An upper bound for the expression in \eqref{driftPenaltyExpre} is shown below
	\begin{align}\nonumber
	&\Delta(\bm{\Theta}(t)) + V\sum\limits_{r \in \mathcal{R}}\mathbb{E}\{f_r(t)| \bm{\Theta}(t)\} \leq B \\\nonumber &+\sum\limits_{i\in \mathcal{N}}\mathbb{E}\{X_{i}(t)(p_i(t) - \gamma_i)|\bm{\Theta}(t)\}\\\nonumber
	&+ \sum\limits_{r\in\mathcal{R}}\mathbb{E}\{Z_{u}(t)(\delta_u-\mu_u(t))|\bm{\Theta}(t)\}\\ &+V\sum\limits_{r\in\mathcal{R}}\mathbb{E}\{f_r(t)|\bm{\Theta}(t)\}   \text{,}
	\label{DriftBound}
	\end{align}
	where $B<\infty$ and $B\geq\frac{1}{2}\sum\limits_{i \in\mathcal{N}} \mathbb{E}\{p_i^2(t)+\gamma^2_i(t)|\bm{\Theta}(t)\} + \frac{1}{2}\sum\limits_{r\in\mathcal{R}}\mathbb{E}\{\delta^2_r + \mu_r^2(t)|\bm{\Theta}(t)\} + \frac{1}{2}\sum\limits_{r\in\mathcal{R}}\mathbb{E}\{\alpha_r^2(t) + \mu^2_r(t)|\mathbf{\Theta}(t)\}$.
	The complete derivation of the above bound can be found in Appendix A.
	
	\subsection{Min-Drift-Plus-Penalty Algorithm}
	We observe that the power allocation decision at each time slot does not affect the value of $B$. The minimum \ac{DPC} algorithm observes the virtual queue backlogs of the virtual queues, the actual queue, and the channel states and makes a control action to solve the following optimization problem
	\begin{subequations}
		\label{DPPOpt} 
		\begin{align}\nonumber
		\min\limits_{\bm{p}(t)} \quad &\sum\limits_{i\in \mathcal{N}}X_{i}(t)(p_i(t) - \gamma_i) + \sum\limits_{r\in\mathcal{R}}Z_{u}(t)(\delta_u-\mu_u(t))\\
		&+V\sum\limits_{r\in\mathcal{R}}f_r(t) \label{objDPP}\\
		\text{s.~t.}\quad &\bm{p}(t) \in \mathcal{P}(t) \label{PowerConsDPP} \text{.}
		\end{align}
	\end{subequations}
	\begin{lemma}
		The optimal solution to problem  \eqref{DPPOpt} minimizes the upper bound of the drift-plus-penalty expression given in the right-hand-side of \eqref{DriftBound}.
	\end{lemma}
	\begin{proof}
		See Appendix C.
	\end{proof}
	\begin{theorem}[Optimality of DPC algorithm and queue stability]
		The \ac{DPC} algorithm guarantees that the virtual and the actual queues are strongly stable and therefore, according to Lemma \ref{Th:rateStable}, the time average constraints in \eqref{PowerConstrEq} and \eqref{ThrConstEq} are satisfied. In particular, the time average expected value of the queues is bounded as
		\begin{align}\nonumber
		&\lim _{t \rightarrow \infty} \frac{1}{t} \sum_{\tau=0}^{t-1}\left(\sum_{i \in \mathcal{N}} \mathbb{E}\left\{X_{i}(t)\right\} + \sum_{u \in \mathcal{U}} \mathbb{E}\left\{Z_{u}(t)\right\}\right) \leq\\
		&\frac{B+V\left(f^{*}(\epsilon)-f^{\text{opt}}\right)}{\epsilon}\text{.}
		\end{align}
		In addition, the expected time average of function $f(t)$ is bounded as
		\begin{align}
		\lim_{t \rightarrow \infty} \sup \frac{1}{t} \sum_{\tau=0}^{t-1} \mathbb{E}\{f(\tau)\} \leq f^{\text {opt }}+\frac{B}{V}\text{.}
		\end{align}
		\label{Theorem:OptimalDPP}
	\end{theorem}
	\begin{proof}
		See Appendix D.
	\end{proof}
	
	We summarize the steps of the \ac{DPC} algorithm that solves the power control problem in \eqref{DPPOpt}
	in Algorithm $1$. 
	
	\setlength{\textfloatsep}{0pt}
	\begin{algorithm}[!h]
		\nonumber
		\caption{DPC}\label{alg:MyAlgorithm}
		Input constant $V$,
		Initialization: $X_{i}(0)=0\text{, } \gamma_{i}\text{, }\forall i \in \mathcal{N}$, $Z_u(0)=0$, $\forall u \in \mathcal{U}$.\\
		\For{$t=1:\ldots$}{
			$MinObj\leftarrow \infty$\\
			\For{$i=1:(|\mathcal{N}|+1)$}{
				$p_{i}(t) \in \mathcal{P}(t)$, 
				Calculate $f_{r}(t)\text{, } \forall r \in \mathcal{R}$\\
				$Obj\leftarrow V\sum\limits_{r\in \mathcal{R}}f_{r}(t) + \sum\limits_{r\in\mathcal{R}}X_{r}(t)(p_i(t)-\gamma_r)$ + $\sum\limits_{u\in \mathcal{U}}Z_u(t)(\delta_u-\mu_{i}(t))$\\
				\If{MinObj$>$Obj}
				{$\mathbf{p}'(t)\leftarrow \mathbf{p}(t)$\\
					$MinObj\leftarrow Obj$}
			}
			$\mathbf{p}(t)\leftarrow \mathbf{p}'(t)$\\
			$X_{j}(t+1)\leftarrow \max\left[X_{j}(t)-\gamma_{j}, 0\right] + p_{j}(t)\text{, } \forall j \in \mathcal{N}$\\
			$Z_u(t+1) = \max\left[Z_u(t)-\mu_u(t),0\right] +\delta_u\text{, } \forall u \in \mathcal{U}$
		}
	\end{algorithm}
	In step $1$, we initialize the importance factor $V$ and the length of virtual queues $X_i(0)$, $\forall i \in\mathcal{N}$, and $Z_u(0)$, $\forall u \in \mathcal{U}$. We try all the possible power allocations from the set $\mathcal{P}$ in step $4$, and we find the corresponding value of the objective function in step $6$.  In step $7$, we check if the candidate power allocation gives smaller value of the objective so far. In steps $11-12$, we updated the virtual queues. After the search, we obtain the solution, $\bm{p}'(t)$, for which the objective function takes its minimum value, $MinObj$.

	\section{Simulation Results}
	\begin{figure}[t!]
		\begin{subfigure}{.5\textwidth}
			\centering
			\includegraphics[scale=0.45]{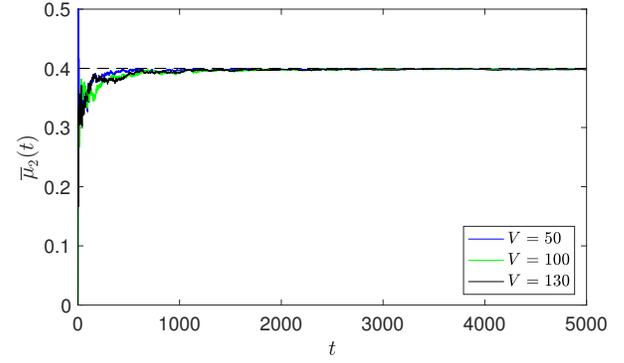}
			\caption{Throughput-constraints convergence of user $2$.}
			\label{Fig:Thrconvergence}
		\end{subfigure}%
		
		\centering
		\begin{subfigure}{.5\textwidth}
			\centering
			\includegraphics[scale=0.45]{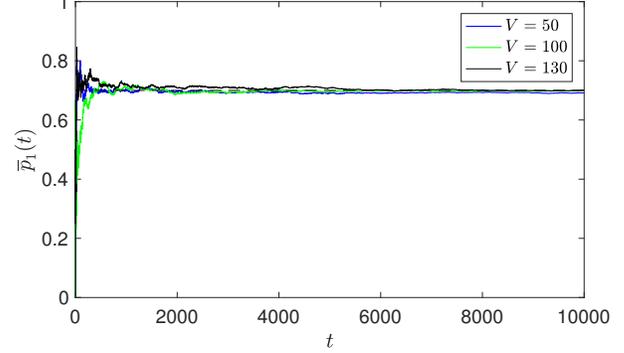}
			\caption{Average power consumption convergence of user $1$.}
			\label{Fig:PowerConvergence}
		\end{subfigure}
		
		\centering
		\begin{subfigure}{.5\textwidth}
			\centering
			\includegraphics[scale=0.45]{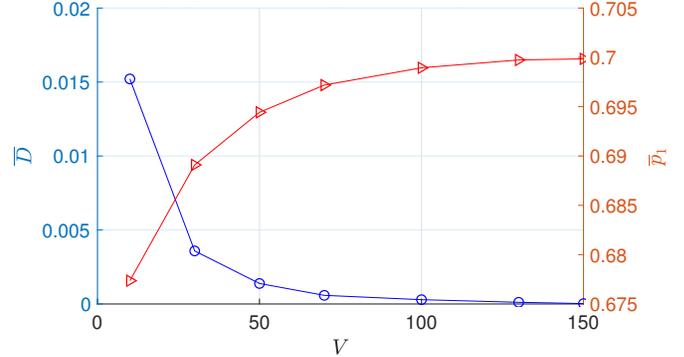}
			\caption{Trade-off between packet drop rate and average power consumption of user $1$ for different values of $V$.}
			\label{Fig:DropVvaluesPower}
		\end{subfigure}
		\caption{\ac{DPC} algorithm performance for different values of $V$. $\lambda_1=0.5$, $\gamma_1=0.7$, $\gamma_2=0.65$, $\delta_2=0.4$, $m=10$.}
		\label{Fig:TradeOffDropsPower}
	\end{figure}
	In this section, we provide results that show the performance of the \ac{DPC} algorithm regarding the packet drop rate and the convergence time for the time average constraints, i.e, throughput, and average power consumption constraints. We first show the results of a system with two users. Our goal is to show the performance of \ac{DPC} for different values of the importance factor $V$. Second, we compare our proposed algorithm with a baseline algorithm called \ac{LDF} algorithm proposed in \cite{TheoryQoS2009}.
	
	\subsection{Performance and Convergence of \ac{DPC} algorithm for different values of $V$}
	In Fig. \ref{Fig:TradeOffDropsPower}, we provide results for a system with two users; user $1$: deadline-constrained user, user $2$: user with minimum-throughput requirements. The average power thresholds are $\gamma_1=0.7$ and $\gamma_2=0.65$ for user $1$ and user $2$, respectively. The minimum-throughput requirements for user $2$ is $\delta_2=0.4$ packets/slot, and the deadlines of the packet of user $1$ is $m=10$ slots. The probability the channel to be in ``Good" state and in ``Bad" state is $0.4$ and $0.6$, respectively. The high level power and the low level power is $P^{\text{High}}=2$ power units and $P^{\text{Low}}=1$ power units, respectively.
	
	In Fig. \ref{Fig:Thrconvergence}, we show the convergence of the algorithm regarding the minimum-throughput constraints for different values of the importance factor $V$. We observe that as the value of $V$ increases, the time convergence increases as well. However, we observe that after approximately $2500$ slots, the algorithm converges and the minimum-throughput requirements are satisfied. In Fig. \ref{Fig:PowerConvergence}, we provide results for the converge of the \ac{DPC} algorithm regarding the average power consumption constraints of user $1$. The algorithm converges after approximately $8000$ slots for each value of $V$.  The probability the channel of user $1$ to be in ``Good" state  is $0.4$. Therefore, the user needs to transmit with high power level for a large portion of the time and that affects the average power consumption. Therefore, the average power consumption constraint with  $\gamma_2=0.7$ is a tight constraint and the algorithm needs more time to converge.
	
	In Fig. \ref{Fig:DropVvaluesPower}, we show the trade-off between packet drop rate and average power consumption of user $1$. We observe that as the value of $V$ increases the average power consumption increases and approaches threshold $\gamma_1$. For $V=10$, we observe that the average power consumption is far from the threshold. In this case, the value of the virtual queue that corresponds to the average power consumption is larger than the cost function for large period of time because the importance factor is relatively small. Therefore, the \ac{DPC} seeks to minimize the larger term of the objective function that is the value of the virtual queue. On the other hand, as we increase $V$, the cost function is weighted more and therefore, the $\ac{DPC}$ algorithm seeks to minimize the cost function which is the most weighted term in the objective function. However, the average power consumption remains always below  threshold  $\gamma_1$.
	
	\subsection{\ac{DPC} vs \ac{LDF}}
	\begin{figure}[t!]
		\begin{subfigure}{.5\textwidth}
			\centering
			\includegraphics[scale=0.45]{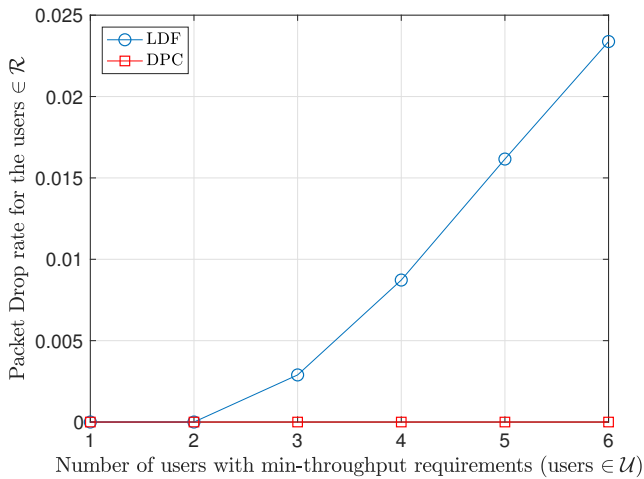}
			\caption{Packet deadline: $m=10$.}
			\label{fig:dropratelambda035m10}
		\end{subfigure}
		\centering
		\begin{subfigure}{.5\textwidth}
			\centering
			\includegraphics[scale=0.45]{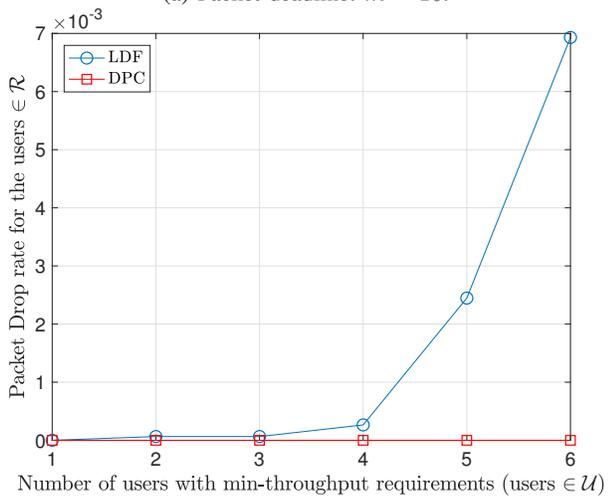}
			\caption{Packet deadline: $m=30$.}	
			\label{fig:dropratelambda035m30}
		\end{subfigure}
		\caption{\ac{DPC} vs \ac{LDF}. Drop rate comparison for different values of  packet deadline $m$.}
		\label{fig:dropratelambda035}
	\end{figure}
	
	\begin{figure}[t!]
		\begin{subfigure}{.5\textwidth}
			\centering
			\includegraphics[scale=0.45]{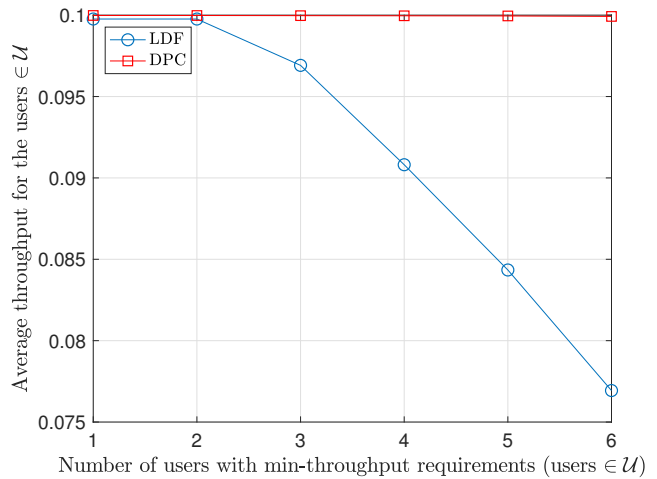}
			\caption{Packet deadline: $m=10$.}
			\label{fig:throughputlambda035m10}
		\end{subfigure}
		\centering
		\begin{subfigure}{.5\textwidth}
			\centering
			\includegraphics[scale=0.45]{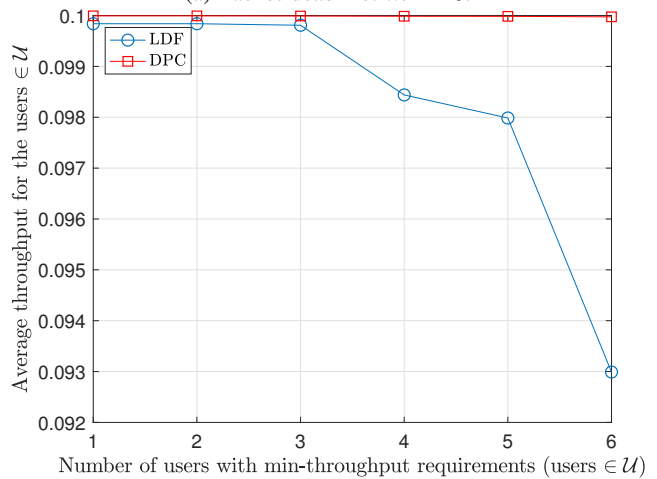}
			\caption{Packet deadline: $m=30$.}	
			\label{fig:throughputlambda035m30}
		\end{subfigure}
		\caption{\ac{DPC} vs \ac{LDF}. Throughput comparison for different values of packet deadline $m$.}
		\label{fig:throughputlambda035}
	\end{figure}
	In this subsection, we compare the performance of our algorithm with that of \ac{LDF}. The \ac{LDF} algorithm allocates power to the user with the \textit{largest-debt}. Algorithm \ac{LDF} selects, at each time slot $t$, the node with the highest value of $y_i(t)$, where $y_i(t)$ is the throughput debt and is defined as
	\begin{align}
	y_i(t+1) = tq_i - \sum\limits_\tau^t \mu_{i}(t) \label{eq:throughputdebt}\text{,}
	\end{align}
	where $q_i$ is the throughput requirements for user $i$. In our case, for the users with throughput requirements, $q_i=\delta_i$. For the deadline-constrained users, $q_i$ will be equal to the percentage of the desired served packets for user with deadlines. For example, if our goal is to achieve zero drop rate we set $q_i=\lambda_i$. However, this is not always feasible, i.e., zero drop rate and satisfaction of the throughput constraints. Therefore, in this case, we get higher drop rate and lower throughput. Note that the \ac{LDF} algorithm does not account for the average power constraints. It was shown in \cite{TheoryQoS2009} that \ac{LDF} is throughput optimal when the problem is feasible for systems with users with throughput requirements. 
	
	In this set up, we consider one user with packets with deadlines and a set with multiple users with minimum-throughput requirements. The probability the channel to be in ``good" state is equal to $0.9$ for all the users. In order to observe a fair comparison between the algorithm, we consider that the average power threshold is $2$ for all the users. Therefore, the average power constraint for every user is always satisfied. The arrival rate for user $1$ is $\lambda_1=0.35$ packets/slot.
	
	In Fig. \ref{fig:dropratelambda035}, we compare the performance of the algorithms regarding the packet drop rate as the number of users with minimum-throughput requirements increases. In Fig. \ref{fig:dropratelambda035m10}, the deadline for the packets of user $1$ is $m=10$. We observe that the \ac{DPC} algorithm outperforms the \ac{LDF} algorithm in terms of drop rate as the number of minimum-throughput requirements increases. In Fig. \ref{fig:dropratelambda035m30}, the deadline for the packets of user $1$ is $m=30$. We observe that the performance of \ac{LDF} has been improved. However, the \ac{DPC} outperforms \ac{LDF} in this case as well. We observe that \ac{LDF} is more sensitive on the size of the deadline of the packets. 
	
	In Fig. \ref{fig:throughputlambda035}, we compare the performance of algorithm regarding the average total throughput of users with minimum-throughput requirements. In Fig. \ref{fig:throughputlambda035m10}, we show results for packets deadline that is $m=10$. Also in this case, we observe that the \ac{DPC} algorithm outperforms the \ac{LDF}. However, for larger deadlines, the performance of the \ac{LDF} is improved, as shown in Fig. \ref{fig:throughputlambda035m30}.

	\begin{figure}[t!]
		\centering
		\includegraphics[scale=0.5]{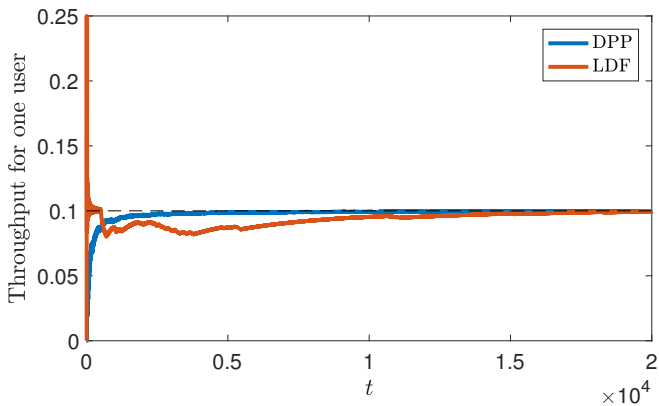}
		\caption{\ac{DPC} vs \ac{LDF}. Min-throughput requirements convergence. Packets deadline: $m=100$ slots.}
		\label{fig:convergencem100}
	\end{figure}
	In Fig. \ref{fig:convergencem100}, we provide results for the convergence time of throughput requirements of one user. In this set up, we consider a system with one user with packet with deadlines and six users with minimum-throughput requirements. In the previous results that are shown in Figs. \ref{fig:throughputlambda035} and \ref{fig:dropratelambda035}, we observe that as we increase the value of $m$, we get a better performance of the \ac{LDF} algorithm. In the system of which the results are shown in Fig. \ref{fig:convergencem100}, we set a very large value to the deadlines of user $1$. 
	We observe that the \ac{DPC} algorithm converges much earlier than \ac{LDF} algorithm. However, both algorithms converge after many slots. This explains the phenomenon of worst performance of \ac{LDF} for small values of $m$. Since the \ac{LDF} algorithm allocates power to the users with largest throughput-debt, it does not consider the remaining slots for each packet. Therefore, when the deadline, $m$, is small, the packets expire before the algorithm converges in terms of throughput requirements  and the drop rate increases.
	
	\section{Conclusions}
	In this work, we consider a system with users with heterogeneous traffic. In particular, the system consists of two sets of users. The first set contains users with packets with deadlines and the second contains users with minimum-throughput requirements, all with limited power budget. All the users transmit their information over a wireless channel by adjusting their power transmission in order to achieve reliable communication. We consider the minimization of the packet drop rate while guaranteeing minimum-throughput. A dynamic algorithm is provided that solves the scheduling problem in real-time. We prove that our scheduling scheme provides a solution arbitrarily close to the optimal. Simulation results show that the proposed algorithm outperforms the baseline algorithm, \ac{LDF}, when the deadlines are short and it is faster in terms of convergence.
	
	\section{Acknowledgments}
	This work has been partially supported by the Swedish Research Council
	(VR) and by CENIIT.
	\begin{appendices}
		\section{Upper Bound on the Lyapunov Drift of DPC}
		Using the fact that \((\max [Q-b, 0]+A)^{2} \leq Q^{2}+A^{2}+b^{2}+\)
		\(2 Q(A-b),\) we rewrite  \eqref{Xevolution}, \eqref{Zevolution}, as 
		\begin{align}
		X_i^2(t+1)  &\leq X^2_i(t) +p^2_i + \gamma_i^2+2X_i(t)(p_i(t)-\gamma_i)\text{,}\label{Xsquare}\\
		Z_u^2(t+1)&\leq Z_u^2(t) +\delta_u^2 + \mu_u^2(t) +2Z_u(t)(\delta_u - \mu_u(t))\label{Zsquare}\text{,}
		\end{align} 
		respectively. Rearranging the terms in \eqref{Xsquare} and \eqref{Zsquare}, dividing them by two, and taking the summations, we obtain
		\begin{align}
		\nonumber
		\sum\limits_{i \in\mathcal{N}}\frac{X_i^2(t+1) - X_i^2(t) }{2} &\leq \sum\limits_{i \in\mathcal{R}} \frac{p_i^2(t)+\gamma_i^2}{2} \label{SumX}\\
		&+ \sum\limits_{i\in\mathcal{N}} X_{i}(t) (p_i(t)-\gamma_i) \\
		\sum\limits_{u\in\mathcal{U}}\frac{Z^2_u(t+1) - Z_u^2(t)}{2} & \leq \sum\limits_{u\in \mathcal{U}}\nonumber \frac{\delta_u^2+\mu_u^2(t)}{2} \\
		&+ \sum\limits_{u\in \mathcal{U}}Z_u(t)(\delta_u - \mu_u(t))\label{SumZ} \text{.}
		\end{align}
		Taking the conditional expectations in  \eqref{SumX} and \eqref{SumZ}, and adding them together, we obtain the bound for the Lyapunov drift in \eqref{DriftBound}. To prove that $B$ that is bounded, we have to find an example and a scheduling scheme for which $B$ takes its maximum value that is bounded. We consider the following set-up:
		\begin{itemize}
			\item The scheduler allocates power at every time slot with the maximum power level,
			\item $\gamma_i=P^\text{High}$, $\forall i \in  \mathcal{N}$,
			\item $\delta_u=1$\text{, } $\forall u \in \mathcal{U}$\text{,}
			\item $\lambda_i= 1$\text{, } $\forall i \in \mathcal{N}$\text{.}
		\end{itemize}
		Then, for the above scheduling scheme, we set $B = \frac{1}{2}\sum\limits_{i \in\mathcal{N}} \mathbb{E}\{p_i^2(t)+\gamma^2_i(t)|\bm{\Theta}(t)\} + \frac{1}{2}\sum\limits_{r\in\mathcal{R}}\mathbb{E}\{\delta^2_r + \mu_r^2(t)|\bm{\Theta}(t)\} + \frac{1}{2}\sum\limits_{r\in\mathcal{R}}\mathbb{E}\{\alpha_r^2(t) + \mu^2_r(t)|\mathbf{\Theta}(t)\} = 
		\frac{1}{2} |\mathcal{N} + 1|(P^{\text{High}})^2 + \frac{1}{2} |\mathcal{R}+1| + \frac{1}{2}|\mathcal{R}+1|<\infty\text{.}$ We observe that even in the scenario in which we take the maximum values of $\gamma_i$, $\delta_u$, and $\lambda_i$, $B$ is bounded.

		\section{Proof of Lemma 1}
		\begin{proof}
			Using the basic sample property \cite[Lemma 2.1, Chapter 2]{NeelyBook}, we have
			\begin{align}
			\frac{X_{i}(t)}{t}-\frac{X_{i}(0)}{t} &\geq  \frac{1}{t} \sum\limits_{\tau = 0} ^{t-1} p_{i}(\tau) - \frac{1}{t}\sum\limits_{\tau = 0} ^{t-1} \gamma_{i}\text{,}\\
			\frac{Z_u(t)}{t} - \frac{Z_u(0)}{t} & \geq \frac{1}{t} \sum\limits_{\tau=0}^{t-1} \delta_u - \frac{1}{t}\sum\limits_{\tau=0}^{t-1}\mu_u(t)\text{.}
			\end{align}
			Therefore, if $X_{i}(t)$ and $Z_{u}(t)$ are rate stable \footnote{A discrete time process $Q(t)$ is strongly stable if \(\limsup _{t \rightarrow \infty} \frac{1}{t} \sum_{\tau=0}^{t-1} \mathbb{E}\{|Q(\tau)|\}<\infty\), \cite{NeelyBook}.} , so that $\frac{X_{i}(t)}{t}\rightarrow 0\text{, } \forall i \in \mathcal{N}$, and $\frac{Z_u(t)}{t}\rightarrow 0$, $\forall u \in \mathcal{U}$, with probability 1, then constraints \eqref{PowerConstrEq} and \eqref{ThrConstEq} are satisfied with probability $1$ \cite{QueueStabilityNeely}.
		\end{proof} 
		\section{Proof of Lemma 2}
		\begin{proof}
			Let $\mathbf{p}(t)$ represent any, possibly randomized, power allocation decision made at slot $t$. Suppose  that $\mathbf{p}^{*}(t)$ is the optimal solution to problem  (\ref{DPPOpt}), and under action $\mathbf{p}^{*}(t)$ the value of $f_{i}(t)$ yields $f_{i}^{*}(t)$ and that of $\mu_{u}(t)$, $\mu^{*}_u(t)$. Then, we have
			
			\begin{small}
				\begin{align}\label{Ineq:OptBound}\nonumber
				&V f^{*}(t)  +  \sum\limits_{i\in\mathcal{N}}X_{i}(t)(p^*(t)-\gamma_i)  + \sum\limits_{u\in\mathcal{U}}Z_u(t)(\delta_u-\mu_u^*(t))\\
				&\leq 
				V f(t)  +  \sum\limits_{i\in\mathcal{N}}X_{i}(t)(p(t)-\gamma_i)  +\sum\limits_{u\in\mathcal{U}}Z_u(t)(\delta_u-\mu_u(t))
				\end{align}
			\end{small}
			Taking the conditional expectations of (\ref{Ineq:OptBound}), we have the result as
			\begin{align}\nonumber
			&V\mathbb{E}\{ f^{*}(t)|\bm{\Theta}(t)\}  +  \sum\limits_{i\in\mathcal{N}}\mathbb{E}\{X_{i}|\bm{\Theta}(t)\}(t)(p_i^*(t)-\gamma_i) \\\nonumber
			&+\sum\limits_{u\in\mathcal{U}}\mathbb{E}\{Z_u(t)|\bm{\Theta}(t)\}(\delta_u-\mu_u^*(t))\leq \\\nonumber
			&V\mathbb{E}\{ f(t)|\bm{\Theta}(t)\}  +  \sum\limits_{i\in\mathcal{N}}\mathbb{E}\{X_{i}|\bm{\Theta}(t)\}(t)(p_i(t)-\gamma_i) \\\nonumber
			&+\sum\limits_{u\in\mathcal{U}}\mathbb{E}\{Z_u(t)|\bm{\Theta}(t)\}(\delta_u-\mu_u(t)) \text{.}
			\end{align}
			
		\end{proof}

		\section{Proof of Theorem 2}
		\begin{proof}
			Suppose that a feasible policy $\omega$ exists, i.e., constraints \eqref{PowerConstrEq} and \eqref{ThrConstEq} are satisfied. Suppose that, for the $\omega$ policy, the followings hold
			\begin{align}
			\mathbb{E}\{p_i(t) - \gamma_i\} &\leq -\epsilon \label{ineqPowerEpsilon}\text{,}\\
			\mathbb{E}\{\delta_u - \mu_u(t)\} &\leq-\epsilon\label{ineqThEpsilon}\text{,}\\
			\mathbb{E} \{f^*(\epsilon)\} & = f^*(\epsilon)\text{,}
			\end{align}
			where $f^*(\epsilon)$ is a sub-optimal solution.
			Applying \eqref{ineqPowerEpsilon} and \eqref{ineqThEpsilon} into \eqref{DriftBound}, we obtain
			\begin{align}\nonumber
			&\mathbb{E}\{L(\bm{\Theta}(t+1)) \}- \mathbb{E}\{ L(\bm{\Theta}(t)) \}+ V\mathbb{E}\{f(t)\}\leq \\\nonumber
			&B -\epsilon \left(\sum\limits_{i\in\mathcal{N}}\mathbb{E}\{X_i(t)\}
			+ \sum\limits_{u\in\mathcal{U}}\mathbb{E}\{Z_u(t)\} \right)+Vf^*(\epsilon)\label{IneqDriftEpsilon}\text{,}
			\end{align}
			taking $\epsilon\rightarrow 0$ and the sum over $\tau=0, \ldots, t-1$ we obtain
			\begin{align}
			\frac{1}{t} \sum\limits_{\tau=0}^{t-1}\mathbb{E}\{f(\tau)\} \leq \frac{ -\mathbb{E}\{L(\bm{\Theta}(t))\} + \mathbb{E}\{L(\bm{\Theta}(0))\} +Bt}{Vt} + f^{\text{opt}} \text{,}
			\end{align}
			taking $t\rightarrow \infty$, we obtain
			\begin{align}
			\lim_{t \rightarrow \infty} \sup \frac{1}{t} \sum_{\tau=0}^{t-1} \mathbb{E}\{f(\tau)\} \leq f^{\text {opt }}+\frac{B}{V}\text{.}
			\end{align}
			That concludes the second part of Theorem \ref{Theorem:OptimalDPP}. In order to prove the stability of the queues, we manipulate \eqref{IneqDriftEpsilon}
			\begin{align}\nonumber
			&\left(\sum_{i \in \mathcal{N}} \mathbb{E}\left\{X_{i}(t)\right\}+\sum_{u \in \mathcal{U}} \mathbb{E}\left\{Z_{u}(t)\right\}\right)\leq\\
			&\frac{B}{\epsilon} - \frac{\mathbb{E}\{L(\mathbf{\Theta}(t+1))\}-\mathbb{E}\{L(\mathbf{\Theta}(t))\}}{\epsilon}  - \frac{V(f^*(\epsilon)   -f(t)    )}{\epsilon}\text{.}
			\end{align}
			By taking the sum over $\tau = 0, \ldots, t-1$ and divide by $t$, we obtain
			\begin{align}\nonumber
			& \frac{1}{t} \sum\limits_{\tau=0}^{t-1} \left(\sum_{i \in \mathcal{N}} \mathbb{E}\left\{X_{i}(t)\right\}+\sum_{u \in \mathcal{U}} \mathbb{E}\left\{Z_{u}(t)\right\}\right)\leq\\
			& \frac{B}{\epsilon} -  \frac{\mathbb{E}\{L(\boldsymbol{\Theta}(t))\}-\mathbb{E}\{L(\boldsymbol{\Theta}(0))\}}{t\epsilon}+\frac{V\left(f^{*}(\epsilon)-f(t)\right)}{\epsilon}\text{,}
			\end{align}
			neglecting the negative term and taking $t\rightarrow \infty$, we have
			\begin{align}\nonumber
			&\lim_{t\rightarrow \infty}\frac{1}{t} \sum\limits_{\tau=0}^{t-1} \left(\sum_{i \in \mathcal{N}} \mathbb{E}\left\{X_{i}(t)\right\}+\sum_{u \in \mathcal{U}} \mathbb{E}\left\{Z_{u}(t)\right\}\right)\leq\\
			&\frac{B+V\left(f^{*}(\epsilon)-f(t)\right)}{\epsilon}\text{.}
			\end{align}
			Consider that $\mathbb{E}\{f(t)\} \geq f^{\text{opt}}$, we obtain the final result as
			\begin{align}\nonumber
			&	\lim _{t \rightarrow \infty} \frac{1}{t} \sum_{\tau=0}^{t-1}\left(\sum_{i \in \mathcal{N}} \mathbb{E}\left\{X_{i}(t)\right\}+\sum_{u \in \mathcal{U}} \mathbb{E}\left\{Z_{u}(t)\right\}\right) \leq\\
			&	\frac{B+V\left(f^{*}(\epsilon)-f^{\text{opt}}\right)}{\epsilon}\text{.}
			\end{align} 
			This shows that the queues are strongly stable for $\epsilon>0$.
		\end{proof}

	\end{appendices}
	
	
	\bibliographystyle{IEEEtran}
	\bibliography{MyBib1}
	
\end{document}